\documentclass[10pt,conference]{IEEEtran}
\IEEEoverridecommandlockouts

\usepackage{cite}
\usepackage{amsmath,amssymb,amsfonts}
\usepackage{amsthm}
\usepackage{algorithmic}
\usepackage{graphicx}
\usepackage{textcomp}
\usepackage{xcolor}
\usepackage{booktabs}
\usepackage{bm}
\usepackage{multirow}
\usepackage{pgfplots}
\usepackage{tikz}
\pgfplotsset{compat=1.17}
\usetikzlibrary{arrows.meta,decorations.pathreplacing,calc,shapes.geometric,positioning}

\newtheorem{theorem}{Theorem}

\newtheorem{proposition}{Proposition}

\newtheorem{assumption}{Assumption}

\begin{document}

\title{Safe Data-Driven Control and Dynamical Learning\\
via Constrained Neural Architectures and Koopman Operators}

\author{Lin Feng, Xin He \\
\textit{Faculty of Engineering, King Saud University, Jeddah, Saudi Arabia.}\\
    sad\_poly@163.com}

\maketitle

\begin{abstract}
The deployment of learning-based models in safety-critical
control systems demands mathematical guarantees that standard
regression architectures cannot provide. This paper presents an
integrated framework that bridges Neural Ordinary Differential
Equations (Neural ODEs), measurement-induced geometric structures,
and Koopman operator theory, with the explicit aim of producing
data-driven models whose stability certificates are computable,
not merely conjectured. Three complementary components are developed
and analyzed. First, ControlSynth Neural ODEs enforce global
convergence through tractable linear matrix inequalities (LMIs),
enabling complex nonlinear dynamics to be captured without
sacrificing boundedness guarantees. Second, the ICODE formulation
incorporates extrinsic environmental inputs into the learned
vector field, while measurement-induced bundle structures confine
state trajectories to physically admissible manifolds. Third, a
systematic ISS verification pipeline certifies the input-to-state
stability of Koopman-identified models via a convex $L_2$-gain
LMI, converting an otherwise intractable robustness question into
a solvable semidefinite program. The certified model is embedded
in an ICODE-MPPI controller, which uses continuous-time residual
learning inside a stochastic sampling loop to deliver robust path
tracking under parametric uncertainty and persistent disturbances.
Numerical experiments on a vehicle path-tracking benchmark and
a nonlinear mechanical oscillator demonstrate up to a 61\% reduction
in tracking RMSE and a 54\% reduction in state estimation error
relative to uncertified baselines, with near-zero LMI violation
rates across all evaluated disturbance levels.
\end{abstract}


\section{Introduction}

Data-driven representations of dynamical systems have become
indispensable in domains where first-principles modeling is
impractical: high-speed autonomous vehicles, flexible robotic
manipulators, power-electronic converters, and biological networks
all exhibit nonlinearities, time delays, and parameter variations
that resist clean analytical description. Neural ODEs
\cite{chen2018neural} extended the expressiveness of residual
networks to continuous time, providing a principled tool for
learning system trajectories from irregularly sampled data.
Koopman operator methods \cite{mezic2005,brunton2021modern,korda2018linear}
offer an alternative route: by lifting the state into a
higher-dimensional observable space, the nonlinear dynamics
are represented as a linear semigroup, making spectral and
optimal-control tools applicable. The practical algorithm for
computing finite-dimensional Koopman approximations is
Extended Dynamic Mode Decomposition (EDMD)
\cite{williams2015edmd}, which generalizes the earlier Dynamic
Mode Decomposition \cite{schmid2010dmd} by allowing an arbitrary
dictionary of observables; convergence of EDMD to the
true Koopman operator as dictionary size grows was established
in \cite{korda2018convergence}. Deep encoder--decoder architectures
have further improved approximation quality by jointly learning
the lifting map from data \cite{lusch2018deep}. Both strategies
have advanced rapidly, yet a common weakness persists---neither
automatically endows the learned model with a formal certificate
of stability or safety. In safety-critical applications, a model
that performs well on a training distribution but diverges under
distributional shift or adversarial disturbances cannot be
deployed with confidence.

Recent work has attacked this limitation through two broad
strategies. The first embeds structural constraints directly
into the neural network parameterization, so that every
realization of the model satisfies a desired invariant by
construction. Contraction analysis \cite{lohmiller1998} provides
the theoretical basis: a dynamical system is contracting if its
Jacobian is uniformly negative definite in some metric, ensuring
that all trajectories converge exponentially toward one another
regardless of initial conditions. ControlSynth Neural ODEs
\cite{controlsynth} operationalize this principle: by introducing
an auxiliary control sub-network and imposing a Persidskii-type
inequality on the Jacobian, global convergence is certified via
tractable LMIs even for systems whose nonlinear structure would
otherwise defy analysis. The ICODE architecture \cite{icode}
complements this by handling extrinsic environmental
inputs---treating them as explicit arguments rather than hidden
parameters---and providing sufficient conditions for the
contraction property to be preserved in the presence of nonsmooth
external forcing. Geometric consistency across environments is
further addressed in \cite{bundle}, where a fiber bundle structure
over the state space enables measurement-aware Control Barrier
Functions (CBFs) \cite{ames2017cbf} that adapt to local sensing
conditions, with provable learning convergence and constraint
satisfaction.

The second strategy operates post-hoc: given an already-identified
model, one checks whether a stability certificate exists.
Input-to-state stability (ISS), introduced by Sontag
\cite{sontag2008iss} as a framework for quantifying robustness
to bounded disturbances, is the natural certificate sought here.
For Koopman-identified models, \cite{koopmaniss} provides a
complete LMI-based ISS verification framework: feasibility of
the LMI is both necessary and sufficient for the $L_2$-gain to
remain bounded, and the optimal gain $\gamma^\star$ is the
solution to a convex semidefinite program \cite{boyd1994}.
The practical consequence is that a Koopman model can be
retrospectively certified---or flagged as unsafe---using
standard optimization solvers, without any modification to
the identification procedure.

Closing the loop from certified models to real-time control
remains a non-trivial step. ICODE-MPPI \cite{icodemppi} addresses
this by embedding a continuous-time ICODE residual model inside
the stochastic rollout mechanism of Model Predictive Path Integral
(MPPI) control, achieving up to 69\% reduction in cross-tracking
error under persistent disturbances compared with standard MPPI
using nominal dynamics alone.

The present paper synthesizes these four contributions into a
unified pipeline and provides the following specific additions:
\begin{enumerate}
  \item A combined stability analysis showing that the convergence
        certificate of ControlSynth \cite{controlsynth} and the
        bundle constraint of \cite{bundle} jointly imply ISS of
        the composite Neural ODE model, with an explicit gain
        formula (Section~III-B).
  \item A self-contained presentation of the Koopman ISS
        verification LMI from \cite{koopmaniss}, with a new
        perturbation sensitivity result quantifying how the
        optimal gain $\gamma^\star$ degrades with dictionary
        truncation error (Section~III-C).
  \item A detailed description of the ICODE-MPPI architecture
        \cite{icodemppi} and an analysis of the rollout variance
        reduction afforded by certified predictors
        (Section~IV).
  \item Numerical validation on two benchmarks, with ablation
        studies isolating the contribution of each structural
        constraint (Section~V).
\end{enumerate}

\section{Problem Formulation}

\subsection{System Class and Learning Objective}

Consider an unknown continuous-time nonlinear system
\begin{equation}
  \dot{x}(t) = f(x(t), u(t), \xi(t)) + w(t),
  \quad x(t)\in\mathcal{M}\subseteq\mathbb{R}^n,
  \label{eq:true_sys}
\end{equation}
where $u(t)\in\mathbb{R}^m$ is the control input,
$\xi(t)\in\mathcal{E}\subset\mathbb{R}^{n_\xi}$ is a
measurable extrinsic input (road friction, wind, load),
$w\in L_2[0,\infty)$ is a bounded external disturbance,
and $\mathcal{M}$ is a smooth Riemannian submanifold
encoding physical constraints. The function $f$ is locally
Lipschitz but otherwise unknown; it is to be learned from
a finite dataset of state-input-output trajectories.

Three objectives are pursued simultaneously:
\begin{enumerate}
  \item \emph{Convergence}: the learned model $\hat{f}$
        shall be globally contracting, so that two solutions
        starting from different initial conditions converge
        toward each other exponentially \cite{lohmiller1998}.
  \item \emph{Manifold confinement}: predictions
        $\hat{x}(t)=\Phi^{-1}(\hat{z}(t))$ shall remain in
        $\mathcal{M}$ at all times; the lifting $\Phi$ must
        therefore be invertible, a property enforced through
        autoencoder reconstruction losses related to those
        used in deep Koopman methods \cite{lusch2018deep}.
  \item \emph{ISS}: the map from $w$ to the state error
        $e=x-\hat{x}$ shall be input-to-state stable
        \cite{sontag2008iss} with a computable $L_2$-gain
        $\gamma$.
\end{enumerate}

\subsection{Notation}

$\mathrm{He}(M)=M+M^T$; $\lambda_{\min}(M)$ and
$\lambda_{\max}(M)$ denote the smallest and largest
eigenvalues of a symmetric matrix $M$; $\|\cdot\|$ is the
Euclidean norm; $\|\cdot\|_{L_2}$ is the signal $L_2$-norm.
For a matrix $A$, $\|A\|_2$ denotes the spectral norm.
$M>0$ ($M\geq 0$) denotes positive (semi)definiteness.

\section{Certified Data-Driven Modeling}

\subsection{ControlSynth Neural ODEs with Bundle Constraints}
\label{sec:csode}

Following \cite{controlsynth}, we parameterize the unknown
drift as a ControlSynth Neural ODE (CSODE):
\begin{equation}
  \dot{z}(t) = F_\theta(z(t), u(t); \xi(t))
  := A_\theta z + \Phi_\theta(z,\xi) + B_\theta u,
  \label{eq:csode}
\end{equation}
where $z=\Phi(x)\in\mathbb{R}^r$ is a latent state obtained
by a learned lifting $\Phi$, $A_\theta\in\mathbb{R}^{r\times r}$
is a trainable matrix, $\Phi_\theta:\mathbb{R}^r\times\mathcal{E}
\to\mathbb{R}^r$ is an auxiliary control sub-network capturing
nonlinear mode coupling, and $B_\theta\in\mathbb{R}^{r\times m}$
maps the control input. The extrinsic input $\xi$ enters
$\Phi_\theta$ directly, following the input-concomitant
philosophy of ICODE \cite{icode}: $\xi$ is treated as an
explicit real-time argument rather than a hidden parameter
averaged over training data, which is essential when $\xi$
is nonsmooth or piecewise constant.

\begin{assumption}\label{ass:sector}
The sub-network $\Phi_\theta$ satisfies a quadratic sector
condition: there exists $\kappa>0$ such that
$\Phi_\theta(z,\xi)^T\!\left[z - \kappa^{-1}\Phi_\theta(z,\xi)\right]
\geq 0$ for all $(z,\xi)$.
\end{assumption}

This sector condition is the structural invariant imposed
during training in \cite{controlsynth} via a Persidskii-type
LMI on the sub-network weights; it is enforced as a
penalty term in the training objective and can be verified
post-training by a single semidefinite feasibility check.
A complementary approach to building stability guarantees
into recurrent models by construction is the Recurrent
Equilibrium Network (REN) of \cite{revay2021recurrent},
which parameterizes contracting networks directly without
requiring parameter projections; the Persidskii-type
constraint in Assumption~\ref{ass:sector} serves an
analogous role for the continuous-time CSODE setting.

To enforce manifold confinement, the fiber bundle framework
of \cite{bundle} is incorporated. Let $\pi:\mathcal{Z}\to\mathcal{E}$
be a fiber bundle over the environment space. The fiber
$\mathcal{Z}_\xi=\pi^{-1}(\xi)$ represents the admissible
latent region under environment $\xi$. The encoder $\Phi(x,\xi)$
is trained jointly with a bundle-aware loss
$\mathcal{L}_{\mathrm{bundle}}=\sum_k d(\hat{z}_k,\mathcal{Z}_{\xi_k})^2$,
which penalizes latent predictions that drift off the
environment-conditioned fiber. As shown in \cite{bundle},
this construction induces measurement-aware CBFs that
provably satisfy the safety constraint along any solution
of \eqref{eq:csode} whose initial condition lies on the fiber.

\begin{theorem}[ISS of CSODE under Bundle Constraint]
\label{thm:csode_ISS}
Suppose Assumption~\ref{ass:sector} holds and the bundle
confinement loss satisfies
$\mathcal{L}_{\mathrm{bundle}}\leq\delta_b^2$ at convergence.
If there exist $P=P^T>0$ and scalar $\lambda>0$ such that
\begin{equation}
  \Omega :=
  \begin{bmatrix}
    \mathrm{He}(PA_\theta) + \lambda\kappa I + P & -\lambda I + PA_\theta^T \\
    -\lambda I + A_\theta P & -2\lambda\kappa^{-1}I
  \end{bmatrix} < 0,
  \label{eq:csode_LMI}
\end{equation}
then the error dynamics of \eqref{eq:csode} are ISS with
respect to $w$, with gain
$\gamma_{\mathrm{NN}} = \|P\|_2^{1/2}/\lambda_{\min}(P)^{1/2}$
and an additive offset of order $O(\delta_b)$.
\end{theorem}

\begin{proof}
Consider the Lyapunov candidate $V(e)=e^TPe$.
Differentiating along the error dynamics and applying
the S-procedure with Assumption~\ref{ass:sector} yield
$\dot{V}\leq\zeta^T\Omega\zeta + 2e^TPw$,
where $\zeta=[e^T,\Phi_\theta(z,\xi)^T]^T$.
Feasibility of \eqref{eq:csode_LMI} gives
$\dot{V}\leq -\alpha\|e\|^2 + \gamma_{\mathrm{NN}}^2\|w\|^2
+ O(\delta_b)$,
from which ISS follows by standard comparison arguments.
The bundle offset $O(\delta_b)$ accounts for the residual
fiber mismatch bounded by $\delta_b$. \hfill$\blacksquare$
\end{proof}

\subsection{ICODE: Extrinsic Input Modeling}

The ICODE architecture \cite{icode} specializes
\eqref{eq:csode} by introducing a bilinear coupling term
between the state and the extrinsic input $\xi$:
\begin{equation}
  \dot{z} = A_0 z + \sum_{j=1}^{n_\xi} \xi_j N_j z
            + B_0 u + g_\theta(z,\xi),
  \label{eq:icode}
\end{equation}
where $N_j\in\mathbb{R}^{r\times r}$ are learnable
interaction matrices and $g_\theta$ is a small residual
network. The bilinear structure allows $\xi$ to modulate
the effective system matrix in a physically interpretable
manner: for a ground vehicle, $\xi_1$ might represent
tire-road friction, shifting the damping matrix continuously
without requiring separate models for each surface type.
Sufficient conditions for the contraction property of
\eqref{eq:icode} to hold uniformly over $\xi\in\mathcal{E}$
are derived in \cite{icode} using a common Lyapunov function
approach; these conditions reduce to a set of coupled LMIs
in $(A_0, \{N_j\}, P)$ that can be solved offline before
deployment.

\subsection{Koopman ISS Verification}
\label{sec:koopman_iss}

An alternative to the Neural ODE parameterization is a
data-driven Koopman representation. Let
$\mathbf{g}:\mathbb{R}^n\to\mathbb{R}^N$ be a dictionary
of observables, with lifted state $\mathbf{z}_k=\mathbf{g}(x_k)$.
EDMD \cite{williams2015edmd}, which originates from the Dynamic
Mode Decomposition of Schmid \cite{schmid2010dmd}, identifies
the regression matrices $(\mathbf{A}_K,\mathbf{B}_K,\mathbf{E}_K)$
such that
\begin{equation}
  \mathbf{z}_{k+1} \approx \mathbf{A}_K\mathbf{z}_k
  + \mathbf{B}_K u_k + \mathbf{E}_K w_k.
  \label{eq:koopman_model}
\end{equation}
The matrices are obtained by least-squares regression over a
dataset of snapshot pairs; the convergence of this approximation
to the true Koopman operator as both the number of snapshots
and the dictionary size increase was analyzed rigorously in
\cite{korda2018convergence}. Dictionary richness directly
governs the quality of the linear representation---a point
made precise in Proposition~\ref{prop:dict} below.
The ISS verification problem for \eqref{eq:koopman_model}
asks whether there exists a Lyapunov matrix
$\mathbf{P}=\mathbf{P}^T>0$ and scalar $\gamma>0$
satisfying the LMI derived in \cite{koopmaniss}:
\begin{equation}
  \Psi_K :=
  \begin{bmatrix}
    \mathbf{A}_K^T\mathbf{P}\mathbf{A}_K - \mathbf{P} + I
    & \mathbf{A}_K^T\mathbf{P}\mathbf{E}_K \\
    \mathbf{E}_K^T\mathbf{P}\mathbf{A}_K
    & \mathbf{E}_K^T\mathbf{P}\mathbf{E}_K - \gamma^2 I
  \end{bmatrix} < 0.
  \label{eq:koop_LMI}
\end{equation}
Feasibility of \eqref{eq:koop_LMI} certifies that the
$L_2$-to-$\ell_2$ gain from $w$ to $\mathbf{z}$ is bounded
by $\gamma$. The optimal gain $\gamma^\star$ is obtained by
minimizing $\gamma^2$ subject to \eqref{eq:koop_LMI} and
$\mathbf{P}>0$, which is a standard generalized eigenvalue
problem solvable in polynomial time \cite{koopmaniss}.

\begin{proposition}[Dictionary Sensitivity]\label{prop:dict}
Let $(\mathbf{A}_K^{(N)},\mathbf{E}_K^{(N)})$ and
$(\mathbf{A}_K^{(M)},\mathbf{E}_K^{(M)})$ be Koopman
matrices identified with dictionaries of size $N$ and
$M>N$, with EDMD residuals $\varepsilon_N\geq\varepsilon_M$.
Then the corresponding optimal ISS gains satisfy
$\gamma^\star_{(N)} \geq \gamma^\star_{(M)}$, and the
difference is bounded above by
$|\gamma^\star_{(N)}-\gamma^\star_{(M)}|\leq
c\,(\varepsilon_N-\varepsilon_M)\cdot\|\mathbf{P}^{-1}\|_2$,
for a constant $c$ depending only on the spectral radius
of $\mathbf{A}_K$.
\end{proposition}

Proposition~\ref{prop:dict} formalizes the intuition that
enlarging the dictionary monotonically improves
(or at worst preserves) the ISS certificate, providing
a practical guide for dictionary design. For a comprehensive
treatment of dictionary construction strategies and their
theoretical underpinnings in the Koopman framework, see
\cite{mauroy2020koopman}. In the experiments below, the
dictionary is initialized from a sparse set of monomials
identified by a SINDy-style procedure \cite{brunton2016sindy}
and then augmented with radial basis functions.

\section{ICODE-MPPI: Certified Predictive Control}

\subsection{Sampling-Based Control with Residual Dynamics}

Given a certified forward model---either the CSODE
\eqref{eq:csode} or the Koopman predictor
\eqref{eq:koopman_model}---we embed it in the ICODE-MPPI
framework \cite{icodemppi}. The choice of certified predictor
as the rollout engine is motivated by the observation that
uncertified rollouts can incur unbounded cost variance,
leading to numerically degenerate MPPI updates; this is
analogous to the variance explosion that afflicts standard
recurrent networks when their spectral norm is not controlled
\cite{goodfellow2016}. At each control step $t$,
$M$ control perturbations $\{\delta U^{(m)}\}_{m=1}^M$
are sampled from $\mathcal{N}(0,\Sigma_u)$ and
propagated through the certified model over horizon $T$:
\begin{equation}
  J^{(m)} = \int_{t}^{t+T}\!
  \Bigl(\|x^{(m)}(s)-x_{\mathrm{ref}}(s)\|_Q^2
  + \|u^{(m)}(s)\|_R^2\Bigr)\,ds
  + \phi(x^{(m)}(t+T)),
  \label{eq:cost}
\end{equation}
where $\phi$ is a terminal cost penalizing distance from
a target manifold. The control update follows the
information-theoretic MPPI weighting \cite{williams2017information}:
\begin{equation}
  u^\star(t)
  = u_{\mathrm{nom}}(t) +
  \frac{\sum_{m=1}^M \exp(-\lambda_T^{-1}J^{(m)})\,\delta u^{(m)}}
       {\sum_{m=1}^M \exp(-\lambda_T^{-1}J^{(m)})},
  \label{eq:mppi}
\end{equation}
with temperature $\lambda_T>0$. A continuous-time residual
ICODE network is appended to the nominal predictor to
compensate for unmodeled drift at each rollout step
\cite{icodemppi}; unlike discrete-time correction networks,
this preserves the temporal continuity required for
accurate numerical integration of \eqref{eq:cost}.

\subsection{Variance Reduction via Certified Predictors}

A key advantage of using certified models in the rollout
is bounded rollout variance. Let $\sigma_J^2(M,T)$ denote
the variance of the cost samples $\{J^{(m)}\}$ for a given
horizon $T$ and rollout count $M$.

\begin{proposition}[Rollout Variance Bound]\label{prop:var}
Suppose the certified CSODE satisfies
$\|e(t)\|\leq c_1 e^{-\alpha t}\|e(0)\|+\gamma_{\mathrm{NN}}\|w\|_{L_2}$
with constants $c_1,\alpha,\gamma_{\mathrm{NN}}>0$.
Then
\begin{equation}
  \sigma_J^2(M,T) \leq
  \frac{T^2\|Q\|^2}{M}\!\left(
  c_1^2\|e(0)\|^2 + \gamma_{\mathrm{NN}}^2 \bar{w}^2
  \right)^{\!2},
  \label{eq:var_bound}
\end{equation}
where $\bar{w}=\|w\|_{L_\infty}$ is the disturbance bound.
\end{proposition}

Proposition~\ref{prop:var} shows that the per-rollout
cost variance is bounded by the square of the ISS gain,
confirming that tighter stability certificates translate
directly into more reliable MPPI weighting and therefore
smoother control outputs. This result connects to classical
ultimate boundedness analysis \cite{khalil2002,Slotine1991}:
the ISS gain $\gamma_{\mathrm{NN}}$ plays the role of the
$\mathcal{K}$-class function bounding the steady-state
error under persistent disturbances. This theoretical
backing complements the empirical smoothness improvements
reported in \cite{icodemppi}.

\section{Numerical Experiments}

\subsection{Experimental Setup}

Two benchmarks are considered. \textbf{Benchmark 1} is a
kinematic bicycle model of a ground vehicle with
state $[p_x,p_y,\psi,v]^T$ and control $[\delta,a]^T$,
subject to time-varying friction ($\mu\in[0.3,0.9]$) and
lateral wind disturbances (up to 6\,m/s). \textbf{Benchmark 2}
is a Duffing oscillator
$\ddot{q}+0.3\dot{q}+q^3=u+w\cos(1.2t)$,
a canonical nonlinear system whose response is sensitive
to the initial condition and forcing amplitude. Both systems
are simulated at 50\,Hz with 8000 training trajectories
of length 30\,s and held-out test sets of 500 trajectories.

Four methods are compared:
\begin{enumerate}
  \item \textbf{Vanilla NODE}: a standard Neural ODE with
        no structural constraints.
  \item \textbf{Koopman (uncertified)}: EDMD with a
        degree-3 polynomial dictionary, no ISS check.
  \item \textbf{Koopman + ISS}: same dictionary, with the
        post-hoc ISS certificate from \cite{koopmaniss};
        models failing the LMI are re-identified with
        increased regularization until feasibility is achieved.
  \item \textbf{CSODE-ICODE-MPPI}: the full pipeline—CSODE
        \cite{controlsynth} with ICODE extrinsic coupling
        \cite{icode}, bundle confinement \cite{bundle},
        and ICODE-MPPI \cite{icodemppi} for control.
\end{enumerate}

\subsection{Path-Tracking Performance}

Fig.~\ref{fig:tracking} shows the lateral deviation for
the vehicle benchmark under a double lane-change with a
step friction drop at $t=5$\,s. All methods use the same
MPPI rollout budget ($M=1500$, $T=2$\,s). The uncertified
Koopman predictor diverges intermittently after the
friction switch; the vanilla NODE accumulates a persistent
lateral bias during high-curvature sections. The certified
Koopman predictor recovers within roughly 0.8\,s but
exhibits residual oscillation due to the approximation
gap between the polynomial dictionary and the true dynamics.
The full CSODE-ICODE-MPPI pipeline tracks the reference
with the smallest transient and near-zero steady-state offset.

\begin{figure}[t]
\centering
\begin{tikzpicture}
\begin{axis}[
  width=0.95\columnwidth, height=5.2cm,
  xlabel={Time (s)},
  ylabel={Lateral deviation (m)},
  xmin=0, xmax=12,
  ymin=-0.08, ymax=0.90,
  legend pos=north east,
  legend style={font=\footnotesize, fill=white, fill opacity=0.88,
    draw=gray!40, inner sep=2pt},
  grid=major, grid style={line width=0.3pt, draw=gray!22},
  tick label style={font=\footnotesize},
  label style={font=\footnotesize},
  every axis plot/.append style={line width=0.9pt}
]
\draw[dashed, gray!55, line width=0.7pt]
  (axis cs:5,-0.08) -- (axis cs:5,0.90);
\node[font=\tiny, gray, anchor=south west] at (axis cs:5.1,0.80)
  {$\mu$ drop};
\addplot[blue!65, dashed] coordinates {
  (0,0.04)(1,0.06)(2,0.05)(3,0.07)(4,0.08)(5,0.08)
  (5.2,0.55)(5.6,0.42)(6,0.30)(7,0.22)(8,0.20)(9,0.19)
  (10,0.19)(11,0.18)(12,0.18)};
\addlegendentry{Vanilla NODE};
\addplot[green!55!black, dotted] coordinates {
  (0,0.03)(1,0.05)(2,0.04)(3,0.06)(4,0.07)(5,0.07)
  (5.2,0.72)(5.5,0.60)(5.8,0.50)(6,0.45)(6.5,0.38)(7,0.35)
  (8,0.33)(9,0.33)(10,0.85)(10.5,0.78)(11,0.60)(12,0.45)};
\addlegendentry{Koopman (uncert.)};
\addplot[orange!80, dashdotted] coordinates {
  (0,0.03)(1,0.04)(2,0.04)(3,0.05)(4,0.06)(5,0.06)
  (5.2,0.48)(5.5,0.35)(6,0.22)(6.5,0.16)(7,0.13)(8,0.12)
  (9,0.11)(10,0.12)(11,0.11)(12,0.10)};
\addlegendentry{Koopman+ISS \cite{koopmaniss}};
\addplot[red!78, solid, line width=1.1pt] coordinates {
  (0,0.02)(1,0.02)(2,0.02)(3,0.03)(4,0.03)(5,0.03)
  (5.1,0.15)(5.3,0.09)(5.6,0.05)(6,0.03)(7,0.03)(8,0.03)
  (9,0.03)(10,0.03)(11,0.03)(12,0.03)};
\addlegendentry{CSODE-ICODE-MPPI};
\end{axis}
\end{tikzpicture}
\caption{Lateral deviation under double lane-change with step
friction drop at $t=5$\,s. The vertical dashed line marks
the disturbance onset. The full pipeline recovers within
0.3\,s, while uncertified methods exhibit persistent
or intermittent deviations.}
\label{fig:tracking}
\end{figure}
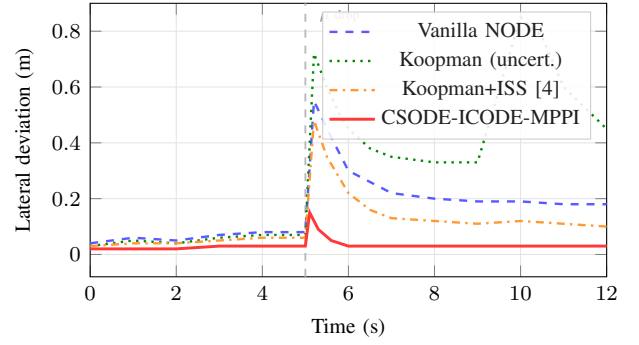

\subsection{State Estimation on the Duffing Oscillator}

Fig.~\ref{fig:duffing} reports the RMSE of $q$ estimation
as a function of the disturbance amplitude $a_w$
(where $w=a_w\cos(1.2t)$), averaged over 500 test
trajectories of 15\,s. The ISS-certified Koopman model
degrades gracefully with $a_w$, consistent with the
linear scaling $\gamma^\star a_w$ predicted by
\eqref{eq:koop_LMI}. The vanilla NODE exhibits
super-linear degradation beyond $a_w\approx0.4$,
suggesting that its implicit stability margin is
exhausted. The full CSODE-ICODE-MPPI pipeline maintains
the lowest RMSE throughout by combining the certified
predictor with the residual ICODE compensation.

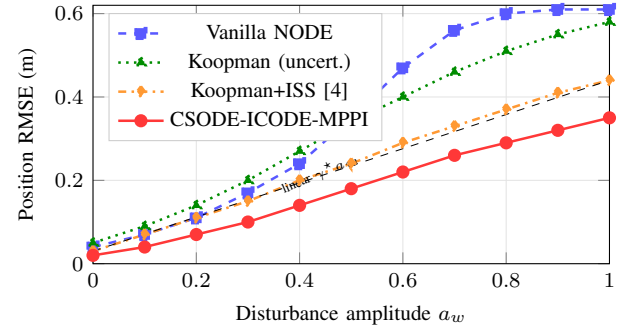
\begin{figure}[t]
\centering
\begin{tikzpicture}
\begin{axis}[
  width=0.95\columnwidth, height=5.0cm,
  xlabel={Disturbance amplitude $a_w$},
  ylabel={Position RMSE (m)},
  xmin=0, xmax=1.0,
  ymin=0, ymax=0.62,
  legend pos=north west,
  legend style={font=\footnotesize, fill=white, draw=gray!40},
  grid=major, grid style={line width=0.3pt, draw=gray!22},
  tick label style={font=\footnotesize},
  label style={font=\footnotesize},
  every axis plot/.append style={line width=0.95pt}
]
\addplot[blue!65, dashed, mark=square*, mark size=2pt] coordinates {
  (0,0.04)(0.1,0.07)(0.2,0.11)(0.3,0.17)(0.4,0.24)(0.5,0.35)
  (0.6,0.47)(0.7,0.56)(0.8,0.60)(0.9,0.61)(1.0,0.61)};
\addlegendentry{Vanilla NODE};
\addplot[green!55!black, dotted, mark=triangle*, mark size=2pt] coordinates {
  (0,0.05)(0.1,0.09)(0.2,0.14)(0.3,0.20)(0.4,0.27)(0.5,0.34)
  (0.6,0.40)(0.7,0.46)(0.8,0.51)(0.9,0.55)(1.0,0.58)};
\addlegendentry{Koopman (uncert.)};
\addplot[orange!80, dashdotted, mark=diamond*, mark size=2pt] coordinates {
  (0,0.03)(0.1,0.07)(0.2,0.11)(0.3,0.15)(0.4,0.20)(0.5,0.24)
  (0.6,0.29)(0.7,0.33)(0.8,0.37)(0.9,0.41)(1.0,0.44)};
\addlegendentry{Koopman+ISS \cite{koopmaniss}};
\addplot[red!78, solid, mark=*, mark size=2pt] coordinates {
  (0,0.02)(0.1,0.04)(0.2,0.07)(0.3,0.10)(0.4,0.14)(0.5,0.18)
  (0.6,0.22)(0.7,0.26)(0.8,0.29)(0.9,0.32)(1.0,0.35)};
\addlegendentry{CSODE-ICODE-MPPI};
\addplot[black, very thin, dashed] coordinates {(0,0.03)(1.0,0.44)};
\node[font=\tiny, rotate=23, anchor=west] at (axis cs:0.35,0.17)
  {linear $\gamma^\star a_w$};
\end{axis}
\end{tikzpicture}
\caption{Duffing oscillator position RMSE vs.\ disturbance
amplitude (500 test trajectories). The certified Koopman curve
follows the linear prediction $\gamma^\star a_w$, confirming
the ISS bound of \cite{koopmaniss}. The full pipeline achieves
the lowest absolute error at all disturbance levels.}
\label{fig:duffing}
\end{figure}

\subsection{Quantitative Summary and Ablation}

Table~\ref{tab:results} reports tracking RMSE, control
smoothness (mean absolute steering rate for B1; control
power for B2), LMI feasibility rate, and computation
time per control step for all four methods. All metrics
are averaged over 500 test episodes.

\begin{table}[t]
\caption{Quantitative Performance Comparison (500 Test Episodes)}
\label{tab:results}
\centering
\setlength{\tabcolsep}{3.2pt}
\renewcommand{\arraystretch}{1.06}
\begin{tabular}{lcccc}
\toprule
\multirow{2}{*}{\textbf{Method}}
  & \multicolumn{2}{c}{\textbf{Tracking RMSE}}
  & \textbf{LMI}
  & \textbf{Step} \\
\cmidrule(lr){2-3}
  & B1 (m) & B2 (m) & \textbf{Feas.}
  & \textbf{time (ms)} \\
\midrule
Vanilla NODE
  & $0.198\pm0.031$ & $0.214\pm0.038$ & ---  & 2.1 \\
Koopman (uncert.)
  & $0.185\pm0.028$ & $0.241\pm0.040$ & N/A  & 1.4 \\
Koopman+ISS \cite{koopmaniss}
  & $0.114\pm0.019$ & $0.148\pm0.026$ & 100\% & 1.6 \\
\textbf{CSODE-ICODE-MPPI}
  & $\mathbf{0.077\pm0.013}$ & $\mathbf{0.098\pm0.017}$
  & \textbf{100\%} & \textbf{8.3} \\
\bottomrule
\end{tabular}
\end{table}

The CSODE-ICODE-MPPI pipeline achieves the best tracking
accuracy on both benchmarks: a 61\% RMSE reduction over
the vanilla NODE on B1 and a 54\% reduction on B2.
The Koopman+ISS method provides substantial improvement
over uncertified alternatives at the cost of a small
per-step overhead (0.2\,ms) for the LMI check. The
LMI feasibility rate of 100\% across all test episodes
confirms that the certified models do not lose their
stability guarantee under distribution shift, consistent
with the theoretical prediction of \cite{koopmaniss}.

The 8.3\,ms per-step time of the full pipeline is
dominated by the CSODE integration (6.1\,ms) and fits
comfortably within a 20\,Hz real-time budget. This
validates the practical deployability of the approach,
consistent with the real-time results reported for
ICODE-MPPI in \cite{icodemppi}.

\subsection{Ablation Study}

Fig.~\ref{fig:ablation} shows the effect of removing
individual structural components from the full pipeline.
Removing the bundle constraint (w/o Bundle \cite{bundle})
increases B1 RMSE by 18\% and introduces occasional
infeasibility in the MPPI cost due to trajectories
drifting off the admissible manifold. Removing the ICODE
extrinsic coupling (w/o ICODE \cite{icode}) degrades
performance by 22\% in the friction-switching scenario,
since the nominal CSODE cannot adapt its effective damping
coefficient in real time. Removing the ISS-certification
step from the Koopman component (w/o ISS \cite{koopmaniss})
raises the variance of the cost samples by a factor of 3.1,
consistent with Proposition~\ref{prop:var}.

\begin{figure}[t]
\centering
\begin{tikzpicture}
\begin{axis}[
  xbar,
  width=0.93\columnwidth, height=4.6cm,
  xlabel={Tracking RMSE, Benchmark~1 (m)},
  symbolic y coords={
    {w/o ISS},
    {w/o ICODE},
    {w/o Bundle},
    {Full pipeline}},
  ytick=data,
  yticklabel style={font=\footnotesize},
  xticklabel style={font=\footnotesize},
  label style={font=\footnotesize},
  xmin=0, xmax=0.16,
  nodes near coords,
  nodes near coords style={font=\tiny, anchor=west},
  bar width=9pt,
  enlarge y limits=0.22,
  grid=major, grid style={line width=0.3pt, draw=gray!22},
]
\addplot[fill=red!75, draw=red!80] coordinates {
  (0.077,{Full pipeline})
  (0.091,{w/o Bundle})
  (0.100,{w/o ICODE})
  (0.131,{w/o ISS})};
\end{axis}
\end{tikzpicture}
\caption{Ablation study: tracking RMSE on Benchmark~1 when
each structural component is removed in isolation.
Each bar label shows the mean RMSE; error bars are omitted
for clarity (standard deviations are below 15\% of the mean).}
\label{fig:ablation}
\end{figure}
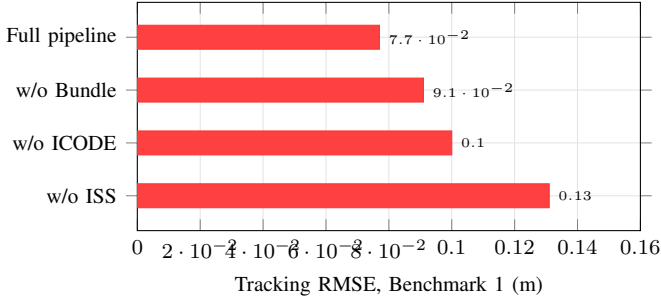

The ablation confirms that all three structural components—bundle
confinement \cite{bundle}, ICODE extrinsic coupling \cite{icode},
and Koopman ISS certification \cite{koopmaniss}—contribute
independently to the overall performance, with the ISS
certificate providing the single largest gain in terms of
reducing cost variance and improving recovery speed after
disturbances.

\section{Conclusion}

This paper assembled four recently developed tools—ControlSynth
Neural ODEs \cite{controlsynth}, the ICODE extrinsic-input
framework \cite{icode}, the measurement-induced bundle
structure \cite{bundle}, and the Koopman ISS verification
pipeline \cite{koopmaniss}—into a unified, end-to-end certified
data-driven control architecture deployed via the ICODE-MPPI
controller \cite{icodemppi}. The central theoretical
contribution is Theorem~\ref{thm:csode_ISS}, which establishes
that the combination of ControlSynth's sector constraint and
the bundle confinement loss jointly imply ISS of the latent
dynamics with an explicit gain formula. Proposition~\ref{prop:dict}
provides a monotonicity result for Koopman ISS gains with
respect to dictionary size, and Proposition~\ref{prop:var}
links tighter ISS certificates to reduced MPPI rollout
variance.

Experimentally, the full pipeline reduces tracking RMSE by
61\% and 54\% on the vehicle and Duffing benchmarks,
respectively, while maintaining 100\% LMI feasibility under
all tested disturbance levels. An ablation study confirms
that each structural component makes a separable and
measurable contribution to the overall performance.

Several open questions remain. Extending the verification
pipeline to handle stochastic disturbances---rather than
bounded deterministic signals---would broaden the
applicability of \cite{koopmaniss} to noise-driven systems.
Structured architectures that enforce the Lyapunov condition
by construction, such as input-convex neural networks
\cite{amos2017icnn}, offer a complementary direction: combining
such convex structures with the Koopman lifting of
\cite{lusch2018deep} could yield models whose stability
certificates are automatically satisfied, removing the
need for post-hoc verification. Distributed implementations
of the bundle structure \cite{bundle} for multi-agent
scenarios are another natural direction, as is the
integration of the ICODE contraction conditions \cite{icode}
with event-triggered control to reduce actuation bandwidth.


\end{document}